\documentclass{article}
\usepackage{spconf,graphicx}
\usepackage{amsfonts,amsmath,amsthm,amssymb}
\usepackage{mathrsfs}
\usepackage{dsfont}
\usepackage{bm}
\usepackage{appendix}

\usepackage{textcomp}
\usepackage{xcolor}
\usepackage{dsfont}
\usepackage{algorithm,algpseudocode}
\usepackage{algpseudocode}
\usepackage{pgfplots}
\pgfplotsset{compat=newest}
\usepackage{subfigure}
\usepackage{float}
\usepackage{xparse}
\usepackage[bookmarks=false, hidelinks]{hyperref}
\usepackage{cite}
\usepackage{tikz}
\usetikzlibrary{arrows,decorations.pathreplacing,patterns}

\newtheorem{lem}{Lemma}	

\theoremstyle{definition}

\allowdisplaybreaks 
\usepackage{balance} 

\newtheorem{theorem}{Theorem} 
\newtheorem{remark}{Remark}

\definecolor{darkcerulean}{rgb}{0.03, 0.27, 0.49}
\definecolor{chestnut}{rgb}{0.8, 0.36, 0.36}
\definecolor{airforceblue}{rgb}{0.36, 0.54, 0.66}
\definecolor{cadmiumorange}{rgb}{0.93, 0.53, 0.18}
\definecolor{bleudefrance}{rgb}{0.19, 0.55, 0.91}
\definecolor{carolinablue}{rgb}{0.6, 0.73, 0.89}
\definecolor{blue(ncs)}{rgb}{0.0, 0.53, 0.74}
\definecolor{brown(web)}{rgb}{0.65, 0.16, 0.16}
\definecolor{pearl}{rgb}{0.94, 0.92, 0.84}
\definecolor{burntumber}{rgb}{0.54, 0.2, 0.14}
\definecolor{asparagus}{rgb}{0.53, 0.66, 0.42}
\definecolor{cssgreen}{rgb}{0.0, 0.5, 0.0}
\definecolor{cadmiumgreen}{rgb}{0.0, 0.42, 0.24}
\definecolor{cadmiumorange}{rgb}{0.93, 0.53, 0.18}
\definecolor{amaranth}{rgb}{0.9, 0.17, 0.31}
\definecolor{bluegray}{rgb}{0.4, 0.6, 0.8}
\definecolor{cadmiumgreen}{rgb}{0.0, 0.42, 0.24}
\definecolor{amaranth}{rgb}{0.9, 0.17, 0.31}
\definecolor{amethyst}{rgb}{0.6, 0.4, 0.8}
\definecolor{amber}{rgb}{1.0, 0.75, 0.0}
\definecolor{azure}{rgb}{0.0, 0.5, 1.0}
\definecolor{babyblue}{rgb}{0.54, 0.81, 0.94}
\definecolor{bazaar}{rgb}{0.6, 0.47, 0.48}
\definecolor{celestialblue}{rgb}{0.29, 0.59, 0.82}
\definecolor{darklavender}{rgb}{0.45, 0.31, 0.59}
\definecolor{bluebell}{rgb}{0.64, 0.64, 0.82}
\definecolor{chamoisee}{rgb}{0.63, 0.47, 0.35}
\definecolor{darkcerulean}{rgb}{0.03, 0.27, 0.49}
\definecolor{iris}{rgb}{0.35, 0.31, 0.81}
\definecolor{dodgerblue}{rgb}{0.12, 0.56, 1.0}
\definecolor{celestialblue}{rgb}{0.29, 0.59, 0.82}
\definecolor{jazzberryjam}{rgb}{0.50, 0.62, 0.37}
\definecolor{cadetgrey}{rgb}{0.57, 0.64, 0.69}

\definecolor{burntsienna}{rgb}{0.81, 0.40, 0.26}
\definecolor{burntumber}{rgb}{0.54, 0.2, 0.14}
\definecolor{bulgarianrose}{rgb}{0.28, 0.02, 0.03}
\definecolor{burgundy}{rgb}{0.5, 0.0, 0.13}
\definecolor{cordovan}{rgb}{0.54, 0.25, 0.27}
\definecolor{eggplant}{rgb}{0.38, 0.25, 0.32}

\definecolor{brickred}{rgb}{0.8, 0.25, 0.33}
\definecolor{chestnut}{rgb}{0.8, 0.36, 0.36}
\definecolor{airforceblue}{rgb}{0.36, 0.54, 0.66}
\definecolor{cadmiumorange}{rgb}{0.93, 0.53, 0.18}
\definecolor{bleudefrance}{rgb}{0.19, 0.55, 0.91}
\definecolor{carolinablue}{rgb}{0.6, 0.73, 0.89}
\definecolor{blue(ncs)}{rgb}{0.0, 0.53, 0.74}
\definecolor{dodgerblue}{rgb}{0.12, 0.56, 1.0}
\definecolor{cssgreen}{rgb}{0.0, 0.5, 0.0}
\definecolor{cadmiumgreen}{rgb}{0.0, 0.42, 0.24}
\definecolor{cadmiumorange}{rgb}{0.93, 0.53, 0.18}
\definecolor{amaranth}{rgb}{0.9, 0.17, 0.31}
\definecolor{bluegray}{rgb}{0.4, 0.6, 0.8}
\definecolor{cadmiumgreen}{rgb}{0.0, 0.42, 0.24}
\definecolor{amber}{rgb}{1.0, 0.75, 0.0}
\definecolor{azure}{rgb}{0.0, 0.5, 1.0}
\definecolor{babyblue}{rgb}{0.54, 0.81, 0.94}
\definecolor{bazaar}{rgb}{0.6, 0.47, 0.48}
\definecolor{celestialblue}{rgb}{0.29, 0.59, 0.82}
\definecolor{darklavender}{rgb}{0.45, 0.31, 0.59}
\definecolor{bluebell}{rgb}{0.64, 0.64, 0.82}
\definecolor{chamoisee}{rgb}{0.63, 0.47, 0.35}
\definecolor{darkcerulean}{rgb}{0.03, 0.27, 0.49}
\definecolor{iris}{rgb}{0.35, 0.31, 0.81}
\definecolor{jazzberryjam}{rgb}{0.65, 0.04, 0.37}
\definecolor{charcoal}{rgb}{0.21, 0.27, 0.31}

\begin{document}

\title{ Optimal QAM Constellation for Over-the-Air Computation \\ in the Presence of Heavy-Tailed Channel Noise
}

\name{Saeed Razavikia$^{\dagger~\star}$,  Deniz Gündüz$^\star$, Carlo Fischione$^\dagger$ \thanks{Emails: \{sraz, carlofi\}@kth.se,
d.gunduz@imperial.ac.uk \\S. Razavikia was supported by the Wallenberg AI, Autonomous
Systems and Software Program.}}
\address{$^\dagger$ KTH Royal Institute of Technology, Stockholm, Sweden\\
$^\star$ Imperial College London, London, UK}

\maketitle

\begin{abstract}
Over-the-air computation (OAC) enables low-latency aggregation over multiple-access channels (MACs) by exploiting the superposition property of the wireless medium to compute functions efficiently in distributed networks. A critical but often overlooked challenge is that electromagnetic interference in practical radio channels frequently exhibits heavy-tailed behavior, causing strong impulsive noise that severely degrades computation performance. This work studies digital OAC with QAM-based signaling under heavy-tailed interference modeled by a Cauchy distribution (lacking a finite second moment). We seek QAM-like constellations that minimize the mean-squared error (MSE) of sum aggregation subject to an average-power constraint. The problem is formulated as a constrained optimization, whose solution yields unique optimality conditions. Numerical results confirm the effectiveness of the proposed design. Notably, the framework extends naturally to nomographic functions, broader constellation families, and alternative noise models. 
\end{abstract}


\vspace{-10pt}
\section{Introduction}

The upcoming 6G networks aim to enable edge intelligence for innovative applications such as augmented reality and the metaverse~\cite{he2025integrating}. As data volumes grow while devices remain resource-limited, computation tasks are increasingly offloaded to edge servers over wireless links. This makes the communication layer, and in particular aggregation protocols, critical to avoid performance bottlenecks~\cite{gunduz2021communicate,perez2025waveforms}.

Over-the-air computation (OAC) addresses this challenge by exploiting the superposition property of multiple-access channels (MACs). By employing simple precoding, concurrent transmissions are superimposed at the receiver to directly obtain the desired aggregate (e.g., sum or mean), thereby lowering both latency and energy consumption compared with conventional transmit-then-aggregate approaches~\cite{nazer2007Computation}. In turn, OAC narrows the communication–computation gap and facilitates applications such as federated learning, distributed inference, and wireless control~\cite{amiri2020federated,yilmaz2025private,park2021optimized}.

Despite its appeal, OAC is typically realized via analog amplitude modulation, which limits compatibility with existing wireless stacks and makes the aggregate highly sensitive to channel noise and fading~\cite{csahin2023survey}. Specifically, in practical MACs, electromagnetic and impulsive interference further induce non-Gaussian, heavy-tailed disturbances \cite{middleton2007statistical}: rare but large outliers dominate the analog superposition, causing severe distortion and bias in the computed function, and leading to unstable updates, e.g., gradient explosion, in federated edge learning\cite{chen2023edge}. Consequently, pure analog OAC becomes a reliability bottleneck, motivating the development of robust aggregation and modulation strategies to sustain performance under realistic MAC impairments.

A growing line of work moves edge aggregation from analog OAC to digital modulation to alleviate noise sensitivity. Early approaches adopt simple constellations (e.g., BPSK/FSK) and recover functions via symbol-type histograms under the type-based multiple-access channel~\cite{zhu2020one,csahin2023over,qiao2024massive,mergen2006type}. Building on this direction,~\cite{saeed2023ChannelComp} proposed a general digital-modulation framework for computing arbitrary finite functions. More recently, SumComp~\cite{Razavikia2024Ring} was introduced as a low-complexity scheme for sum computation, using a two-dimensional integer grid compatible with standard constellations such as multi-level and hexagonal QAM. Its reliability has been further enhanced by incorporating channel coding~\cite{liu2025digital,yan2025remac}.

Most existing studies adapt standard digital modulation schemes, e.g., PAM or QAM, leaving unresolved the fundamental challenge of constellation design tailored for OAC. While constellation optimality has been extensively studied in conventional communications~\cite{makowski2006optimality}, systematic insights for OAC remain scarce, as the objective shifts from message decoding to function estimation. To address this gap, we develop QAM-like constellations optimized to minimize the mean-squared error (MSE) of sum computation under an average-power constraint and heavy-tailed (Cauchy) noise. The complex Cauchy model captures impulsive outliers and reflects interference-limited operation~\cite{yang2021revisiting}. The resulting optimality conditions yield a coupled nonlinear system of equations, for which we characterize the optimality regions, and establish uniqueness in the large network. Finally, we validate our theoretical findings with numerical experiments. We note that the proposed framework naturally extends to the general class of \emph{nomographic} functions and alternative noise models.

\vspace{-15pt}
\section{System Model}
\label{sec:system}
\vspace{-7.5pt}

 We consider $K$ transmitters, and a single receiver, referred to as the computation point (CP).  Each transmitter node owns an integer value   $s_k\in \{0,1,\ldots, Q-1\}$, where $Q$ denotes the size of the alphabet. Then, all the nodes transmit simultaneously over a MAC to enable the CP to compute a desired function $f(s_1,\ldots,s_K)$.  Throughout the paper,   we consider $f$ to be the sum function, i.e., 
\vspace{-8pt}
\begin{align}
    f(s_1,\ldots,s_K) = \sum\nolimits_{k=1}^{K}s_k, 
\end{align}
which allows us to leverage waveform superposition on the MAC. 

\vspace{-5pt}
\subsection{Multiple Access Channel (MAC)}

The transmitter at node $k$ employs the encoder $\mathscr{E}_q(\cdot)$\footnote{Since the target function is symmetric with respect to its inputs, using an identical encoder across nodes suffices for computation~\cite{saeed2023ChannelComp}.} with parameter $q \in \mathbb{Z}^{+}$ to map its input $s_k$ to a channel symbol $x_k$, i.e., $x_k = \mathscr{E}_q(s_k)\in \mathbb{C}$. The channel symbols are drawn from a discrete square constellation of size $q \times q$, which requires $Q = q^2$ with $q \in \mathbb{Z}^{+}$. To satisfy the average-power constraint, the expected symbol energy must not exceed $P$, i.e., $\mathbb{E}[|x_k|^2] \leq P$. Under equiprobable signaling, this condition reduces to  $ \frac{1}{Q}\sum_{m=1}^{Q} |c_m|^2 \leq P$, where $\{c_1,\ldots,c_Q\}\subset \mathbb{C}$ denotes the constellation points.  

All nodes transmit their encoded symbols simultaneously over the shared channel\footnote{Residual synchronization errors at the receiver can be mitigated using phase-coded pilots~\cite{sahin2025feasibility}.}. After perfect channel inversion, the CP observes
\vspace{-15pt}
\begin{align}
\label{eq:fading}
    r = \sum\nolimits_{k=1}^{K} x_k + z,
\end{align}
where $r$ denotes the received signal and $z \in \mathbb{C}$ is additive noise. In contrast to the conventional Gaussian assumption, which may be optimistic in interference-limited scenarios~\cite{clavier2020experimental}, we model $z$ as a complex Cauchy random variable with scale parameter $\gamma > 0$, i.e., $z \sim \mathcal{C}(0,\gamma)$. Explicitly, $z = z_1 + i z_2$, where $z_1$ and $z_2$ are independent and identically distributed as $\text{Cauchy}(0,\gamma)$, thereby capturing impulsive heavy-tailed behavior. Since the Cauchy distribution lacks a finite second moment, it effectively models strong outliers. The CP then recovers the target function value via a decoding map $\mathscr{D}: \mathbb{C} \mapsto \mathcal{Y}_f$, yielding the estimate $\hat{f} := \mathscr{D}(r)$, where $\mathcal{Y}_f$ denotes the output alphabet of the desired function $f$.
  
\vspace{-4pt}
 
 \subsection{Decoding Procedure}

To recover the function estimate $\hat{f}$ from the received signal $r$, we first project $r \in \mathbb{C}$ onto the superimposed symbol grid, which forms a two-dimensional square constellation of size $N \times N$, where $N = (q-1)K+1$. Let $\mathcal{Y}$ denote the set of all induced constellation points with cardinality $|\mathcal{Y}| = N^2$. Then, the maximum-likelihood (ML) decoder is defined as
\vspace{-4pt}
\begin{align}
    \label{eq:ml_decoder}
    \mathscr{D}(r) = \underset{y_j \in \mathcal{Y}}{\arg\max}~ g(r|y_j),
\end{align}
\vspace{-2pt}
where $g(r|y_j)$ is the conditional channel transition probability. Since the noise is Cauchy distributed, we have
\vspace{-4pt}
\begin{align}
    g(r|y) = \frac{\gamma}{\pi \left(\gamma^2 + |r-y|^2\right)}, \quad \gamma > 0. 
\end{align}
Owing to the symmetry of the Cauchy distribution and the independence of the real and imaginary noise components, the two-dimensional ML decoder in \eqref{eq:ml_decoder} decouples into two one-dimensional estimators. Consequently, $\mathscr{D}(\cdot)$ reduces to independently rounding the real and imaginary parts of $r$ to the nearest grid points on the constellation.

\begin{figure}[t]
    \centering

\definecolor{eggshell}{rgb}{0.94, 0.92, 0.84}
\definecolor{flavescent}{rgb}{0.97, 0.91, 0.56}
\definecolor{lightapricot}{rgb}{0.99, 0.84, 0.69}
\definecolor{peach-orange}{rgb}{1.0, 0.8, 0.6}
\definecolor{peach-yellow}{rgb}{0.98, 0.87, 0.68}
\tikzset{every picture/.style={line width=0.75pt}} 
\scalebox{0.65}{
\begin{tikzpicture}[x=0.75pt,y=0.75pt,yscale=-1]

\draw[draw opacity=0][fill=pearl , rounded corners=25pt] (40pt, 200pt) rectangle (210pt, 10pt) {};

\draw [color={rgb, 255:red, 155; green, 155; blue, 155 } ] [dash pattern={on 4.5pt off 4.5pt}]  (60,145) -- (270,145) ;
\draw [color={rgb, 255:red, 155; green, 155; blue, 155 }] [dash pattern={on 4.5pt off 4.5pt}]  (170,260) -- (170,45) ;

\draw (170,30) node   {Gray Code};

\newcounter{ga}\setcounter{ga}{0}

\foreach \y in {145pt,100pt,55pt,10pt} {

           \foreach \x in {10pt,55pt,100pt,145pt} {
            
          \draw[fill=black!90] (50pt+\x,25pt+\y) node{}  circle  (3.5);
           

            \stepcounter{ga};
          }

}

\draw (80,185) node {\footnotesize $0001$};
\draw (80,245) node  {\footnotesize $0000$};
\draw (80,65) node   {\footnotesize $0010$};
\draw (80,125) node  {\footnotesize $0011$};

\draw (140,65) node  {\footnotesize $0110$};
\draw (140,125) node {\footnotesize $0111$};
\draw (140,185) node   {\footnotesize $0101$};

\draw (140,245) node  {\footnotesize $0100$};

\draw (200,125) node  {\footnotesize $1111$};
\draw (260,125) node {\footnotesize $1011$};
\draw (260,65) node   {\footnotesize $1010$};
\draw (200,65) node   {\footnotesize $1110$};

\draw (200,245) node   {\footnotesize $1100$};
\draw (260,245) node  {\footnotesize $1000$};
\draw (260,185) node   {\footnotesize $1001$};
\draw (200,185) node  {\footnotesize $1101$};

\draw[draw opacity=0][fill=pearl, rounded corners=25pt] (255pt, 218pt) rectangle (410pt, 10pt) {};

\draw [color={rgb, 255:red, 155; green, 155; blue, 155 } ] [dash pattern={on 4.5pt off 4.5pt}]  (340,155) -- (540,155) ;
\draw [color={rgb, 255:red, 155; green, 155; blue, 155 }] [dash pattern={on 4.5pt off 4.5pt}]  (440,250) -- (440,45) ;

\draw (445,30) node   {SumComp Code};

\draw [color={rgb, 255:red, 74; green, 144; blue, 226 }] [dash pattern={on 0.84pt off 2.51pt}]  (375,55) -- (515,120) ;

\draw [color={rgb, 255:red, 74; green, 144; blue, 226 } ] [dash pattern={on 0.84pt off 2.51pt}]  (375,120) -- (520,190) ;
\draw [color={rgb, 255:red, 74; green, 144; blue, 226 }] [dash pattern={on 0.84pt off 2.51pt}]  (375,190) -- (515,255) ;

\newcounter{ca}\setcounter{ca}{0}

\newcommand{\yTop}{65pt}   
\newcommand{\ySep}{50pt}    

\foreach \i in {0,1,2,3}{

  \pgfmathsetlengthmacro{\y}{\yTop - \i*\ySep}

  \draw[-latex,color={rgb,255:red,74; green,144; blue,226}]
    (375,\y+125pt) -- (530,\y+125pt);

  \foreach \x in {40pt,75pt,110pt,145pt}{
    \draw[fill=black!90] (240pt+\x,125pt+\y) node{} circle (3.5);
    \draw (240pt+\x,135pt+\y) node{\footnotesize $\theca$};
    \stepcounter{ca};
  }
}

 \draw[latex-latex, color={rgb, 255:red, 78; green, 43; blue, 20 }] (380,75) -- (415,75);
\draw (400,65) node [font=\footnotesize]  {\color{rgb, 255:red, 78; green, 43; blue, 20 }$d_1$};

\draw[latex-latex, color={rgb, 255:red, 78; green, 43; blue, 20 }] (360,60) -- (360,120);
\draw (350,90) node [font=\footnotesize]  {\color{rgb, 255:red, 78; green, 43; blue, 20 }$d_2$};

\begin{scope}[shift={(0cm,0.77cm)}]

\draw (395,255) node  [font=\footnotesize,color={rgb, 255:red, 74; green, 144; blue, 226 } ]  {$1$};
\draw (350,200) node [font=\footnotesize, color={rgb, 255:red, 74; green, 144; blue, 226 }]  {$4$};

\draw [-latex, color={rgb, 255:red, 74; green, 144; blue, 226 }]   (375,225) .. controls (355,205) and (360,180) .. (375,155) ;

\draw [-latex, color={rgb, 255:red, 74; green, 144; blue, 226 }]   (375,225) .. controls (380,245) and (405,250) .. (420,225) ;

\end{scope}

\end{tikzpicture}}

    \caption{Gray code vs SumComp code  for QAM $Q=16$ modulation. The right constellation diagram uses spacing parameters $(d_1^{*},d_2^{*})$ are determined by Theorem~\ref{TH:MAIN}. 
    }
    \label{fig:GrayCode}
\end{figure}
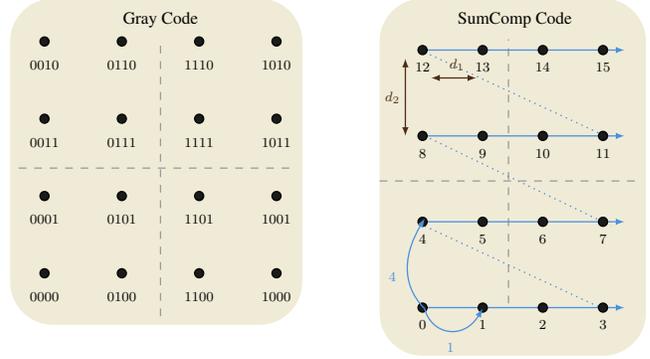

 \vspace{-5pt}

\subsection{Encoding Procedure}
 \vspace{-5pt}
Since the target function is the sum, the encoder $\mathscr{E}_q(s)$ must preserve the additive structure of the input symbols $s$ (forming an additive group). Accordingly, for any $s \in \{0,1,\ldots,Q-1\}$, we define $\mathscr{E}_q(\cdot)$ with parameter $q \in \mathbb{Z}^{+}$ as
\begin{align}
 \label{eq:encoding_qam}
\mathscr{E}_q(s) := \big(s - q \lfloor s/q \rfloor\big)d_1  
  + \lfloor s/q \rfloor d_2 {\rm i}, \quad d_1,d_2 \in \mathbb{R}^+,
\end{align}
where ${\rm i}$ is the imaginary unit, and $d_1, d_2$ specify the spacing along the in-phase and quadrature components, respectively.

\begin{remark}
The encoding rule in \eqref{eq:encoding_qam} can be viewed as a special case of the SumComp scheme~\cite{Razavikia2024Ring}, since it restricts the constellation to a square grid structure. At the same time, it introduces additional flexibility by parameterizing the in-phase and quadrature spacings through $(d_1,d_2)$. 
\end{remark}

\begin{remark}\label{rem:addative}
 For each symbol $s_k$ at node $k$, the pair $(s_k,\mathscr{E}_q(s_k))$ forms an additive group~\cite{Razavikia2024Ring}. Consequently, the sum $\sum_{k=1}^K s_k$ can be uniquely recovered from the encoded aggregate $\sum_{k=1}^K \mathscr{E}_q(s_k)$. By applying an isomorphic decoding map $\mathscr{D}$, the aggregated codeword can be mapped back to the desired computation $f = \sum_{k=1}^K s_k$.  
\end{remark}

The average symbol power of this two-dimensional grid is given by~\cite{goldsmith2005wireless} $\mathbb{E}_{s}\left[|\mathscr{E}_q(s)|^2\right] = {(Q-1)}(d_1^2 + d_2^2)/6$. An example of the resulting modulation diagram for $q=4$ is shown in Fig.~\ref{fig:GrayCode}.

\vspace{-5pt}

\subsection{Problem Statement}
\vspace{-5pt}
In \eqref{eq:encoding_qam}, the encoder $\mathscr{E}_q(s)$ maps each input symbol onto a two-dimensional QAM-like grid, whose geometry is determined by the spacings $(d_1,d_2)$ along the in-phase and quadrature axes.  Since the quadrature axis is scaled by the factor $q$,  symmetric noise in this direction causes proportionally larger computation errors compared to those along the in-phase axis. Thus, robustness requires a larger quadrature spacing, i.e., $d_2 > d_1$. The key design question is therefore:
\emph{By how much should $d_2$ exceed $d_1$ so that the overall computation error at the CP is minimized?} 
To answer this, we define the MSE of the recovered function $\hat f$ as $\mathcal{J}_q(d_1,d_2) = \mathbb{E}\big[|f - \hat f|^2 \big]$, and aim to find the optimal parameters $(d_1,d_2)$ that minimize this error subject to an average-power constraint, i.e., 
\begin{align}
\label{eq:MSEdefinition}
\min_{d_1, d_2} \quad \mathcal{J}_q(d_1,d_2)
\quad \text{s.t.} \quad {(Q-1)} (d_1^2+d_2^2)/6 = P.
\end{align}
where $P$ is the available power budget.

\begin{remark}
 Unlike Gaussian noise, the Cauchy distribution exhibits heavy tails and lacks a finite variance, which may suggest that MSE-based design is inapplicable. Nevertheless, since symbols are digitally represented and hard decoded, the effective function estimation error remains finite~\cite{chen2023quantizing} 
\end{remark}

\vspace{-7pt}

\section{Optimal Constellation Design}
\vspace{-6pt}

This section aims to derive the optimal in-phase and quadrature spacings $(d_1,d_2)$ of the proposed QAM-like constellation that minimize computation error under heavy-tailed Cauchy noise. The analysis proceeds in three steps: (i) we derive a closed-form expression for $\mathcal{J}_q(d_1,d_2)$ that depends solely on $(d_1,d_2)$, thereby reducing the constellation design task to a two-parameter optimization; (ii) we formulate the Lagrangian of the constrained problem and establish the Karush–Kuhn–Tucker (KKT) conditions, which characterize all stationary points of the Lagrangian; and (iii) we show that, for sufficiently large number of transmitter $K$, the KKT system admits a unique feasible solution, which corresponds to the global minimizer of $\mathcal{J}_q(d_1,d_2)$. Lemma~\ref{LEM_PROOF} presents $\mathcal{J}_q(d_1,d_2)$ in terms of $d_1$ and $d_2$. 

\begin{lem}\label{LEM_PROOF}
For a $K$-user MAC with encoder $\mathscr{E}_q(\cdot)$ in \eqref{eq:encoding_qam}, ML decoder $\mathscr{D}$ in \eqref{eq:ml_decoder}, and Cauchy noise $z \sim \mathcal{C}(0,\gamma)$, assume the induced constellation points of $\sum_k s_k$ are uniformly distributed over $\mathcal{Y}$. Then, the MSE is
\begin{align}
\label{eq:MSEd1q2d2}
\mathcal{J}q(d_1,d_2) = \mu(d_1) + q^2 \mu(d_2),
\end{align}
where $\mu(x) = \frac{2}{\pi} \sum\nolimits_{m=1}^{N-1}
\alpha_m \arctan\Big(\frac{\gamma}{(2m-1)x}\Big)$, with $\alpha_m = 2m-1 + \tfrac{3m(1-m)-1}{N}$ and $N = K(q-1)+1$.
\end{lem}
\begin{proof}
 See Appendix~\ref{ap:lem_proof}. 
\end{proof}
\begin{figure}[!t]
\centering
\subfigure[$K=10$]{
    \label{fig:SUMQAM}
    \begin{tikzpicture} 
    \begin{axis}[
        xlabel={$\gamma^{-1}$ (dB)},
        ylabel={${\rm MSE}(\hat{f})$},
        label style={font=\scriptsize},
        tick label style={font=\scriptsize} , 
        width=0.48\textwidth,
        height=3.9cm,
        xmin=0, xmax=20,
        ymin=1e-1, ymax=8000,
        minor tick num=5,
         ymode = log,
      legend style={nodes={scale=0.65, transform shape}, at={(0.45,0.65)}}, 
        ymajorgrids=true,
        xmajorgrids=true,
        grid style=dashed,
        grid=both,
        grid style={line width=.1pt, draw=gray!15},
        major grid style={line width=.2pt,draw=gray!40},
    ]
     \addplot[ color=darkcerulean,
        line width=1pt,
        each nth point={10}
        ]        
    table[x=xi_dB,y=mse_opt]
    {Data/Computation_exp/Cauchy_MSE_q4_K10.dat};
     \addplot[ color=darkcerulean,
        line width=1pt,
        each nth point={10}, 
        dashed
        ]        
    table[x=xi_dB,y=mse_eq]
    {Data/Computation_exp/Cauchy_MSE_q4_K10.dat};
     \addplot[color=bulgarianrose,
     each nth point={10}, 
        line width=1pt]
    table[x=xi_dB,y=mse_opt]
    {Data/Computation_exp/Cauchy_MSE_q8_K10.dat};
     \addplot[color=bulgarianrose, dashed,
     each nth point={10}, 
        line width=1pt]
    table[x=xi_dB,y=mse_eq]
    {Data/Computation_exp/Cauchy_MSE_q8_K10.dat};
    \legend{{$(q=4)$-$({d}_1^*,{d}_2^*)$},{$(q=4)$-$({d}_1={d}_2)$},{$(q=8)$-$({d}_1^*,{d}_2^*)$} ,{$(q=8)$-$({d}_1={d}_2)$},};
    \end{axis}
\end{tikzpicture}
}
\vspace{-5pt}
\subfigure[$K=100$ ]{
    \label{fig:SUMPAM}
    \begin{tikzpicture} 
    \begin{axis}[
        xlabel={$\gamma^{-1}$ (dB)},
        ylabel={${\rm MSE}(\hat{f})$},
        label style={font=\scriptsize},
        tick label style={font=\scriptsize} , 
        width=0.48\textwidth,
        height=3.9cm,
        xmin=0, xmax=20,
        ymin=5e-1, ymax=10000,
        minor tick num=5,
         ymode = log,
      legend style={nodes={scale=0.65, transform shape}, at={(0.45,0.65)}}, 
        ymajorgrids=true,
        xmajorgrids=true,
        grid style=dashed,
        grid=both,
        grid style={line width=.1pt, draw=gray!15},
        major grid style={line width=.2pt,draw=gray!40},
    ]
     \addplot[ color=darkcerulean,
        line width=1pt,
        each nth point={10}
        ]        
    table[x=xi_dB,y=mse_opt]
    {Data/Computation_exp/Cauchy_MSE_q4_K100.dat};
     \addplot[ color=darkcerulean,
        line width=1pt,
        each nth point={10}, 
        dashed
        ]        
    table[x=xi_dB,y=mse_eq]
    {Data/Computation_exp/Cauchy_MSE_q4_K100.dat};
    \addplot[color=darklavender, each nth point={10}, 
        line width=1pt]
    table[x=xi_dB,y=mse_opt]
    {Data/Computation_exp/Cauchy_MSE_q8_K100.dat};
    \addplot[color=darklavender,dashed,
        line width=1pt,
        each nth point={10}]
    table[x=xi_dB,y=mse_eq]
    {Data/Computation_exp/Cauchy_MSE_q8_K100.dat};
    \legend{{$(q=4)$-$({d}_1^*,{d}_2^*)$},{$(q=4)$-$({d}_1={d}_2)$} ,{$(q=8)$-$({d}_1^*,{d}_2^*)$},{$(q=8)$-$({d}_1={d}_2)$}};
    \end{axis}
\end{tikzpicture}}
  \caption{Monte Carlo evaluation of the MSE for the sum function with \ref{fig:SUMQAM} $K=10$ and $K=100$ transmitter over $5\times10^4$ independent trials: Solid curves denote the optimized distance parameters $({d}_1^*,{d}_2^*)$ obtained by, whereas dashed curves correspond to equal‐distance ${d}_1={d}_2=\sqrt{6/(Q-1)}$.
  }
  \label{fig:SumSim}
\end{figure}
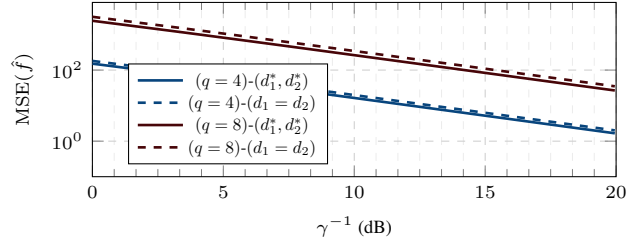
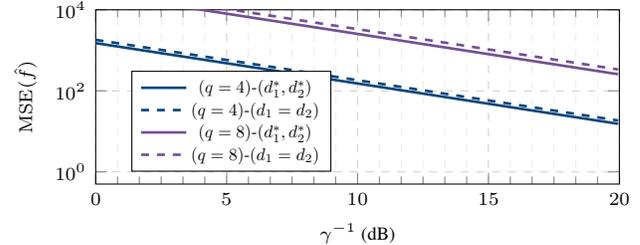

Given the closed-form representation in \eqref{eq:MSEd1q2d2}, the constellation design problem in \eqref{eq:MSEdefinition} reduces to
\begin{align}
\label{eq:msed1d2}
           \min_{d_1, d_2} \; \mu(d_1) + q^2 \mu(d_2),
           \quad \text{s.t.} \quad d_1^2+d_2^2 = \rho^2,
\end{align}
where $\rho := \sqrt{6P/(Q-1)}$. To solve \eqref{eq:msed1d2}, we form the Lagrangian
\begin{align}
\label{eq:Lagrangian_function}
 \mathcal{L}(d_1, d_2, \lambda) 
 = \mu(d_1) + q^2 \mu(d_2) + \lambda ( d_1^2 + d_2^2 - \rho^2 ),
\end{align}
where $\lambda$ is the Lagrange multiplier enforcing the power constraint. Applying the KKT conditions to \eqref{eq:Lagrangian_function}, and using the bordered Hessian theorem~\cite{d2001bordered}, yields the following result.

\begin{theorem}\label{TH:MAIN}
The optimal constellation parameters of \eqref{eq:MSEdefinition} are
\begin{align}
    d_1^\ast =  \gamma \rho \sqrt{0.5 - t^\ast}, 
    \qquad  
    d_2^\ast = \gamma \rho \sqrt{0.5 + t^\ast},
\end{align}
where $t^\ast$ is the unique single positive  root of
\begin{align}
    \mathcal{G}_{Q,\gamma}^N(t) &=
\sum\nolimits_{m=1}^{N-1}
\frac{\theta_m^2}{\sqrt{0.5-t}\big(1+\theta_m\rho^2(0.5-t)\big)}  \nonumber \\
&\quad -
\sum\nolimits_{m=1}^{N-1}
\frac{Q \theta_m^2}{\sqrt{0.5+t}\big(1+\theta_m\rho^2(0.5+t)\big)},
\end{align}
with $\theta_m=2m-1$ for $m\in [M]$, $\mathcal{G}_{Q,\gamma}^N(t^\ast)=0$ for $K \gg 1$.
\end{theorem}

\begin{proof}
    See Appendix~\ref{ap:proof}.
\end{proof}


\section{Numerical Results}
\vspace{-5pt}
Here, we validate the proposed constellation design and illustrate its performance gains over standard QAM schemes under heavy-tailed noise. We evaluate the MSE of the corrupted function with respect to the SNR $P/\gamma$ ($P=1)$, and compare with one archived by the symmetric constellation, i.e.,  $d_1=d_2=\sqrt{6/(Q-1)}$ ($P=1$) employed by the SumComp~\cite{Razavikia2024Ring}. Fig.~\ref{fig:SumSim} plots the resulting MSE as a function of $1/\gamma \in \{0,\dots,20\}\,\mathrm{dB}$ in two scenarios: 
1) $K=100$ nodes with $q \in \{4,8\}$ to illustrate the effect of modulation order; 2) $K=10$ nodes with $q \in \{4,8\}$ to highlight the impact of constellation size.  
Across a broad SNR range, the optimized design consistently yields lower MSE than the QAM-style grid. and the performance gap widens as either $K$ or $q$ increases. As $\gamma$ decreases, all curves converge since the Cauchy distortion becomes negligible and both designs coincide. Overall, the optimized constellation achieves an improvement of approximately $4$ to $5$ dB in MSE across a wide SNR range. Finally, when $q=4$, the optimized design delivers the best performance due to its balanced parameterization.

\vspace{-5pt}
\section{Conclusions}
\label{sec:conclusions}
\vspace{-5pt}
We studied digital OAC over a MAC with QAM constellation family under heavy-tailed Cauchy noise. The constellation design was cast as an optimization problem aiming to minimize the MSE of sum aggregation subject to a power constraint. By analyzing the KKT conditions, we showed that the optimal constellation parameters correspond to the unique root of a nonlinear equation. The resulting design is inherently asymmetric, with $d_2^{*} > d_1^{*}$, thereby quantifying the imbalance between in-phase and quadrature spacings. Numerical results demonstrated up to $4$ dB improvement in MSE compared to conventional SumComp constellations. The proposed framework can be readily extended to nomographic functions, providing a robust and efficient solution for OAC in practical wireless networks.

\vspace{-6pt}
\appendix

\section{Proof of Lemma~\ref{LEM_PROOF}}\label{ap:lem_proof}
\vspace{-5pt}
The proof follows similar arguments as in \cite[Appendix B]{Razavikia2024Ring}. Let $z_1$ and $z_2$ denote the real and imaginary components of the channel noise $z$, i.e., $z = z_1 + z_2 i$. Since $z_1$ and $z_2$ are independent, the MSE decomposes as
\begin{align} \mathcal{J}_q(d_1,d_2) = \mathbb{E}\Big[\|{\mathscr{D}(z_1)}|^2\Big] + q^2 \mathbb{E}\Big[\,|{\mathscr{D}(z_2)}|^2 \Big]. \end{align}
where the first and second terms denote the effective decision errors along the in-phase and quadrature axes, respectively. Also, $|\mathscr{D}(z_1)|$ and $|\mathscr{D}(z_2)|$ are integer-valued random variables representing symbol detection errors in each axis. The expectations $\mu(d_1)$ and $\mu(d_2)$ correspond to the average squared error of an $N$-ary PAM constellation, where $N = K(q-1) + 1$ is the number of superimposed constellation points. Their closed-form expressions can be obtained analogously to \cite[Eq. 42]{Razavikia2024Ring}, by replacing the Gaussian $Q$-function with the tail distribution of the Cauchy density. This yields the expression stated in Lemma~\ref{LEM_PROOF}

\vspace{-5pt}
\section{Proof Sketch of Theorem~\ref{TH:MAIN}} \label{ap:proof}
\vspace{-5pt}

By applying the KKT conditions to the Lagrangian in \eqref{eq:Lagrangian_function}, and after some algebraic manipulations, we obtain
\vspace{-5pt}
\begin{subequations}
\label{eq:KKT_system}
    \begin{align}
 \sum_{m=1}^{N-1} \frac{\gamma_{m}}{d_1\big(1+(\theta_{m}d_1)^2\big)} 
 &= \sum_{m=1}^{N-1} \frac{q^2\gamma_{m}}{d_2\big(1+(\theta_{m}d_2)^2\big)}, \label{eq:KKT_eq1}\\
 d_1^2 + d_2^2 &= \rho^2, \label{eq:KKT_eq2}
    \end{align}    
\end{subequations}
where $\gamma_{m} = \alpha_m(2m-1)$ and $\theta_m = (2m-1)/\gamma$ for all $m \in [N]$.  
The solutions to \eqref{eq:KKT_system} correspond to the stationary points of the Lagrangian $\mathcal{L}(d_1,d_2,\lambda)$. Since $\gamma_m$ may take negative values for $m > \lceil 2N/3 \rceil$, the system may in general admit multiple solutions.

For a large number of nodes $K \gg 1$, however, we can approximate $\gamma_m \approx (2m-1)^2 > 0$. In this case, the system simplifies to
\vspace{-5pt}
\begin{subequations}
\label{eq:KKT_theta}
    \begin{align}
 \sum_{m=1}^{N-1} \frac{\theta_m^2}{d_1\big(1+(\theta_m d_1)^2\big)} 
 &= \sum_{m=1}^{N-1} \frac{q^2\theta_m^2}{d_2\big(1+(\theta_m d_2)^2\big)}, \label{eq:KKT_theta1}\\
 d_1^2 + d_2^2 &= \rho^2. \label{eq:KKT_theta2}
    \end{align}    
\end{subequations}
To eliminate the constraint \eqref{eq:KKT_theta2}, we define an auxiliary variable $t$ such that $d_1 = \rho \sqrt{0.5 - t}$ and $ d_2 = \rho \sqrt{0.5 + t}$ with $t\in (-0.5, 0.5)$. Substituting $t$ into \eqref{eq:KKT_theta1}, we obtain 
\begin{align}
\nonumber
     \mathcal{G}_{Q,\gamma}^{N}(t) &:= 
\sum\nolimits_{m=1}^{N-1}  
\frac{\theta_m^2}{\sqrt{0.5-t}\big(1+\theta_m\rho^2(0.5-t)\big)} \\
&\quad - \sum\nolimits_{m=1}^{N-1} 
\frac{q^2\theta_m^2}{\sqrt{0.5+t}\big(1+\theta_m\rho^2(0.5+t)\big)}. \label{eq:G_function}
\end{align}
Thus, the optimal solution corresponds to the unique root $t^\ast$ of $\mathcal{G}_{Q,\gamma}^{N}(t)$, yielding
\[
d_1^\ast = \rho \sqrt{0.5 - t^\ast}, \qquad d_2^\ast = \rho \sqrt{0.5 + t^\ast}.
\]

It remains to show that $t^\ast$ is unique. Differentiating \eqref{eq:G_function} with respect to $t$ gives
\begin{align}
  \nonumber
   \frac{\partial \mathcal{G}_{Q,\gamma}^{N}(t)}{\partial t} 
&= \sum_{m=1}^{N-1} 
\frac{\theta_m^2}{2(0.5 - t)^{3/2}}
\cdot \frac{1+3\theta_m \rho^2(0.5 - t)}
{\big(1+\theta_m \rho^2(0.5 - t)\big)^2} \\
&\quad + \sum_{m=1}^{N-1} 
\frac{q^2\theta_m^2}{2(0.5+t)^{3/2}}
\cdot \frac{1+3\theta_m \rho^2(0.5+t)}
{\big(1+\theta_m \rho^2(0.5+t)\big)^2}. \nonumber
\end{align} 
 The derivative is strictly positive  for $t\in (0,0.5)$, implying that $\mathcal{G}_{Q,\gamma}^{N}(t)$ is strictly increasing and hence injective. By the intermediate value theorem, and using the fact that $\lim_{t \to +\infty} \mathcal{G}_{Q,\gamma}^{N}(t) = +\infty$ while $\mathcal{G}_{Q,\gamma}^{N}(t)$ is continuous, the function is also surjective. Moreover, by the bordered Hessian (second–order sufficient) test~\cite{d2001bordered}, strictly positive implies that any feasible solution $t^\ast$ of $\mathcal{G}_{Q,\gamma}^{N}(t)=0$ is therefore a \emph{strict local minimizer}. Therefore, because the feasible set is the convex interval $(0,0.5)$ and the stationary point is unique, this local minimizer is in fact the global minimizer.

\bibliographystyle{ieeetr}
\bibliography{IEEEabrv,Ref2}


\end{document}